\documentclass[journal]{IEEEtran}
\usepackage{amsmath,amsfonts,amsthm}
\usepackage{array}
\usepackage[caption=false,font=normalsize,labelfont=sf,textfont=sf]{subfig}
\usepackage{textcomp}
\usepackage{stfloats}
\usepackage{url}
\usepackage{verbatim}
\usepackage{graphicx}
\usepackage{xcolor}
\usepackage{nomencl}
\usepackage{etoolbox}
\usepackage{ifthen}
\usepackage{booktabs}
\usepackage[ruled]{algorithm2e}
\usepackage[normalem]{ulem}

\usepackage{multirow}
\usepackage{bm}
\usepackage{makecell}
\usepackage{enumitem}

\allowdisplaybreaks 
\usepackage[subtle,tracking=normal]{savetrees}

\usepackage{cite}

\def\BibTeX{{\rm B\kern-.05em{\sc i\kern-.025em b}\kern-.08em
    T\kern-.1667em\lower.7ex\hbox{E}\kern-.125emX}}
\usepackage{balance}

\theoremstyle{definition}
\newtheorem{theorem}{Theorem}

\newtheorem{prop}{Proposition}

\newtheorem{corollary}{Corollary}

\newtheorem{remark}{Remark}

\usepackage{url}
\usepackage{hyperref} 
\usepackage{xcolor}

\hypersetup{
    colorlinks=true, 
    linkcolor=blue, 
    filecolor=black, 
    urlcolor=blue, 
    citecolor=blue, 
    linkbordercolor={1 1 1}, 
    pdfborder={0 0 0} 
}
\begin{document}

\bstctlcite{IEEEexample:BSTcontrol}
\title{Chance-Constrained Energy Storage
Pricing \\ for Social Welfare Maximization}

\author{\IEEEauthorblockN{Ning Qi, \textit{Member, IEEE}}, \IEEEauthorblockN{Ningkun Zheng, \textit{Student Member, IEEE}}, \IEEEauthorblockN{Bolun Xu, \textit{Member, IEEE}}

\thanks{This work was partly supported by the Department of Energy, Office of Electricity, Advanced Grid Modeling Program under contract DE-AC02-05CH11231 and partly supported by the National Science Foundation under award ECCS-2239046.
Paper no.TEMPR-XXXXX-2024.

Ning Qi, Ningkun Zheng and Bolun Xu are with the Department of Earth and Environmental Engineering, Columbia University, New York, NY
10027 USA (e-mail: \{nq2176, nz2343, bx2177\}@columbia.edu). }
\vspace{-0.3cm}
}

\markboth{IEEE TRANSACTIONS ON Energy Markets, Policy and Regulation,~Vol.~X, No.~X, XX July~2024}
{How to Use the IEEEtran \LaTeX \ Templates}

\maketitle


\begin{abstract}

This paper proposes a novel framework to price energy storage in economic dispatch with a social welfare maximization objective. This framework can be utilized by power system operators to generate default bids for storage or to benchmark market power in bids submitted by storage participants.
We derive a theoretical framework based on a two-stage chance-constrained formulation which systematically incorporates system balance constraints and uncertainty considerations. We present tractable reformulations for the joint chance constraints. Analytical results show that the storage opportunity cost is convex and increases with greater net load uncertainty. We also show that the storage opportunity prices are bounded and are linearly coupled with future energy and reserve prices. We demonstrate the effectiveness of the proposed approach on an ISO-NE test system and compare it with a price-taker storage profit-maximizing bidding model. Simulation results show that the proposed market design reduces electricity payments by an average of 17.4\% and system costs by 3.9\% while reducing storage's profit margins, and these reductions scale up with the renewable and storage capacity.

\end{abstract}

\section{Introduction}\label{Introduction}

\IEEEPARstart{E}{fficient} management of energy storage resources is critical to reliable and economical operations as their market share continues to surge. The capacity of battery energy storage in the California Independent
System Operator (CAISO) has exceeded 8.6 GW in April 2024 and is projected to reach 50 GW in 2045~\cite{CAISORE}, with most of the storage conducting price arbitrage in wholesale markets~\cite{zheng2023energy}. 
Thus, power system operators must incorporate storage's physical and economic characteristics in market designs to facilitate the convergence of storage participation with social welfare~\cite{williams2022electricity}.

Current market practices primarily rely on storage participants generating strategic bids, with a limited understanding of dispatching storage to minimize system operational costs and maximize social welfare. While traditional electricity markets are based on the marginal fuel cost curve of thermal generators~\cite{kirschen2018fundamentals}, most wholesale markets now allow storage to participate as both generators and responsive demand. In real-time markets, which clear bids up to one hour in advance, storage must account for future opportunity costs in their offers~\cite{baker2023transferable} based on private predictions of future market prices~\cite{ebrahimian2018price,yang2021real}. Consequently, storage operators design these bids to maximize market profits~\cite{krishnamurthy2017energy,qin2024economic}, potentially misaligning social welfare objectives. On the other hand, market power mitigation for storage is challenging due to the inability to distinguish between capacity withholding for capturing legitimate future opportunities and intentions to exercise market power~\cite{xu2024truthful}.

Power system operators must be equipped to generate default storage bids, providing an anchoring point to monitor profit-seeking storage participants.
This paper proposes a novel approach to price energy storage future opportunities for social welfare maximization. The proposed pricing model is designed to run prior to the main economic dispatch, determining the dispatch price for storage.
While previous studies have examined storage opportunity costs from the perspective of profit-seeking participants~\cite{qin2024economic, chen2021pricing}, our paper focuses on social welfare. 
Our contributions are: 
\begin{enumerate}
    \item \textit{Chance-Constrained Opportunity Pricing Framework:} 
        We propose a novel opportunity pricing design for energy storage with a social welfare maximization objective while guaranteeing robust competitive equilibrium. Opportunity price is derived from a chance-constrained two-stage joint economic dispatch for energy and reserve, incorporating uncertainties from renewables and load as well as the risk preference of market operators.  
    \item \textit{Theoretical Pricing Analysis:} 
We provide a theoretical analysis of key characteristics of the proposed pricing framework to establish a robust understanding of market intuitions. We show that the storage opportunity price decreases monotonically with the state of charge (SoC) while increasing monotonically with net load uncertainty. We prove that the opportunity prices have a linear coupling relationship with the energy and reserve prices when storage participates in both energy and reserve markets, and that the opportunity prices have bounded values.
    \item \textit{Simulation Analysis:} 
        We validate our theoretical findings using an 8-zone ISO-NE test system and benchmark it with a price-taker profit-maximizing storage bidding approach. Simulations show that the proposed pricing mechanism significantly reduces system costs and electricity payments compared to existing practices. Furthermore, these reductions scale up with the capacity of storage and intermittent renewables.
\end{enumerate}


We organize the remainder of the paper as follows. 
Section~\ref{LR} summarizes the previous works on market design of energy storage and pricing of uncertainty. Section~\ref{PR}~provides problem formulation and preliminaries of chance-constrained opportunity pricing. Section~\ref{MR}~presents the theoretical pricing analysis. Section~\ref{Case Study}~describes case studies to verify the theoretical results. 
Finally, section~\ref{Conclusion}~concludes this paper.
\vspace{-0.1cm}
\section{Literature Review}\label{LR}

\subsection{Market Design of Energy Storage}

Efforts are underway to adapt electricity market designs to better accommodate special cost components of storage and maximize storage flexibility utilization~\cite{frolke2024efficiency}. Fang et al.~\cite{fang2022efficient} incorporate the cost functions related to cycling mileage and valuation functions of ending SoC into storage bidding. Xu et al.~\cite{xu2020operational} propose an analytical expression for the opportunity cost of storage under multi-stage price uncertainties and prove that it depends on SoC. Building on this, Zheng et al.~\cite{zheng2023energy} proposes a SoC-dependent bidding mechasim while incorporating the opportunity cost of storage. Chen et al.~\cite{chen2023convexifying} incorporates convexified SoC-dependent bids of storage into multi-period economic dispatch and market clearing problem. The above works~\cite{fang2022efficient,zheng2023energy,chen2023convexifying} ensure that storage can bid truthfully if the lost opportunity cost of storage is eliminated~\cite{chen2021pricing}. However, all the mentioned works aim for profit maximization of energy storage with a price-taker assumption, which only benefits large and smart storage participants who are experienced in price prediction and exercising market power. While, the current market design disadvantages system operators and small storage participants who are unwilling or unable to invest extra intelligence to exercise market power. As we move towards a future decarbonized power system with significant amounts of energy storage and renewables, it is crucial for the system operator to develop more systematic methods for storage opportunity price clearing, which aim to maximize social welfare while considering uncertainties.


\subsection{Pricing Uncertainty in Stochastic Electricity Market}

System operator have been rapidly innovating on stochastic electricity market design, while incorporating uncertainties from renewables and load into market pricing and clearing. 
Early works typically model uncertainty using scenario-based stochastic optimization. Seminal work~\cite{wong2007pricing} firstly employed scenario-based stochastic optimization for energy and reserve market clearing. Papavasiliou at al.~\cite{papavasiliou2011reserve} generate a less conservative reserve margin design with lower expected operation cost by stochastic optimization. Qin et al.~\cite{qin2024economic} and Xu et al.~\cite{xu2020operational} incorporate the stochasticity of real-time energy price into energy storage bidding through stochastic dynamic programming. However, market design with stochastic optimization does have certain limitations. First, the revenue adequacy and cost recovery are proved to be satisfied only in expectation, but do not necessarily hold for individual scenarios (e.g., worst-case scenario)~\cite{morales2012pricing}. Although this issue is solved by probability-weighted price~\cite{kazempour2018stochastic}, social welfare maximum and market equilibrium can not guaranteed. Secondly, prices derived from dual variables are generated for each scenario, which complicates market settlements as there is no default pricing scheme~\cite{kuang2018pricing}. Hence, as stated in~\cite{kazempour2018stochastic}, the scenario-based stochastic optimization is intended for market analysis rather than market-clearing tools.

Noticed by these limitations, the chance-constrained approach is proposed which offers superior probabilistic properties compared to stochastic optimization and less conservatism compared to robust optimization. 
Furthermore, chance-constrained optimization can be
exactly reformulated into deterministic expressions and
solved efficiently at scale~\cite{roald2017chance}. And by adjusting the confidence level, the system operator can directly control the constraint violation and guarantee the system reliability. Dvorkin~\cite{dvorkin2019chance} proposes a chance-constrained market design that formulates energy and reserve prices within the unit commitment problem. Dorini et al.~\cite{dorini2013chance} generates the optimal demand response price signals by embedding the finite impulse response model within a chance-constrained framework. Liu et al.~\cite{liu2016distribution} designs a distribution locational marginal pricing for electric vehicle through a chance-constrained approach to alleviate congestion in the distribution network. To the best of our knowledge, no research work has yet designed the opportunity pricing of energy storage for social welfare maximization under the chance-constrained framework.


\section{Problem Formulation and Preliminaries}\label{PR}
We consider a vertically integrated economic dispatch setting where the system operator aims to minimize system operational costs and use the proposed framework to generate opportunity prices for storage. The proposed pricing model operates before the main dispatch/market clearing, serving its results as inputs for the primary market clearing process. In deregulated systems, these pricing results can function as default bids for storage participants or as benchmarks for market power.

     
\subsection{Two-Stage Chance-Constrained Economic Dispatch}
Since we consider a pricing model for social welfare maximization, we simplify our setting to consider a single bus system with an aggregated generator and an aggregated storage system to focus on analyzing the relationship among dispatch decisions, system uncertainties, and prices.
Note that the proposed framework can extrapolate to multiple generators/storages with network models. The formulate is presented as follows:

\noindent\textbf{Functions and Parameters:}
\begin{enumerate}

\item $\mathbb{E}$\text{, }$\mathbb{P}$: Expectation and probability function.
\item ${G}$: Cost function of a conventional generator, which is practically piecewise linear/quadratic and approximately quadratic/super-quadratic.
\item ${M}$: Marginal cost of energy storage.
\item ${{D}_{t}}$\text{, }$\bm{{d}_{t}}$: Forecast value and forecast error of net load at time ${t}$, where net load refers to the difference between load demand and power output of renewables, $\bm{{d}_{t}}$ is a random variable.
\item $\overline{P}$\text{, }$\overline{E}$: Maximum power and energy capacity of energy storage.
\item $\eta$: Charge/discharge efficiency of energy storage.
\item  $\overline{G}$\text{, }$\underline{G}$: Maximum and minimum power output of conventional generator.
\item $\epsilon$: Probability level of chance-constraint, e.g., 5\%.
\item ${\lambda }_{t}$\text{, }${\theta}_{t}$\text{, }${\pi}_{t}$: Energy price, opportunity price and reserve cost at time ${t}$.
\item $\underline{\alpha}_{t}$\text{, }$\overline{\alpha}_{t}$\text{, }$\underline{\beta}_{t}$\text{, }$\overline{\beta}_{t}$\text{, }$\underline{\nu}_{t}$\text{, }$\overline{\nu}_{t}$\text{, }$\underline{\iota}_{t}$\text{, }$\overline{\iota}_{t}$\text{, }$\underline{\kappa}_{t}$\text{, }$\overline{\kappa}_{t}$:  Dual variables of the corresponding constraints at time ${t}$.
\end{enumerate}

\noindent\textbf{Decision Variables:}
\begin{enumerate}
\item ${g}_{t}$: First-stage power output of conventional generator at time ${t}$.
\item $p_{t}$\text{, }$b_{t}$\text{, }${e}_{t}$: First-stage discharge power, charge power and SoC of energy storage at time ${t}$.
\item $\varphi_{t}\text{, }\psi_{t}$: First-stage reserve allocation ratio of conventional generators and energy storage at time ${t}$.
\end{enumerate}
\begin{subequations}
\begin{align}
&  \min\ \mathbb{E}\ \sum\nolimits_{t}\big({{{G}}\left( {{g}_{t}}+\varphi_{t}\bm{{d}_{t}} \right)}+M(p_{t}+\psi_{t}\bm{{d}_{t}})\big)\label{obj}\\
& {g}_{t}+p_{t}-b_{t}={D}_{t}:\text{  }{{\lambda }_{t}} \label{pb}\\
&{e}_{t+1}-{e}_{t}=-p_{t}/{\eta}+b_{t}{{\eta }}:\text{  }{{\theta }_{t}} \label{SoC}\\
&\varphi_{t}+\psi_{t}=1:\text{ }{{\pi}_{t}} \label{rgt}\\
&0\leq\varphi_{t}\text{,}\psi_{t}\leq 1\label{rgre}\\
&\mathbb P\big(\underline{G}\le g_{t}+\varphi_{t}\bm{{d}_{t}}\le \overline{G}\big)\geq 1-\epsilon:\text{ }\underline{\nu }_{t}\text{, }\overline{\nu }_{t} \label{rgulbound}\\
&b_{t}p_{t}=0 \label{st}\\
& 0\le b_{t}\text{, } \mathbb P\big(b_{t}-\psi_{t}\bm{{d}_{t}}\le \overline{P})\geq 1-\epsilon:\text{ }\underline{\alpha}_{t}\text{, }\overline{\alpha }_{t} \label{pcbound}\\
&  0\le p_{t}\text{, } \mathbb P\big(p_{t}+\psi_{t}\bm{{d}_{t}}\le \overline{P})\geq 1-\epsilon:\text{ }\underline{\beta }_{t}\text{, }\overline{\beta }_{t} \label{pdbound}\\
& \mathbb P\big((\psi_{t}\bm{{d}_{t}}+p_{t})/{\eta}\le e_{t}\le \overline{E}-(b_{t}-\psi_{t}\bm{{d}_{t}})\eta\big)\geq 1-\epsilon:\text{ }\underline{\iota}_{t}\text{, }\overline{\iota}_{t}\label{rsulbound}
\end{align}\label{MCCED}
\end{subequations}
\vspace{-0.1cm}

The objective function~\eqref{obj} aims to minimize both the first-stage and second-stage system costs, including the production cost from conventional generators and the degradation cost from energy storage. Constraint~\eqref{pb} guarantees the power balance. Constraint~\eqref{SoC}~defines the relationship between the state of charge (SoC) and charge/discharge actions of storage. Constraints~\eqref{rgt}--\eqref{rgre}~set the reserve allocation ratio to the conventional generator and energy storage. The joint chance-constraints \eqref{rgulbound} limit conventional generators' power and reserve outputs within the power bounds with the confidence level $1-\epsilon$. Constraint~\eqref{st} ensures that storage cannot charge and discharge simultaneously within the same time interval. Charge and discharge power with both energy and reserve provision is limited by \eqref{pcbound} and \eqref{pdbound}, respectively. The joint chance-constraints \eqref{rsulbound} limit the storage SoC within the bounds with the confidence level $1-\epsilon$. For simplicity, generation and storage ramping constraints are omitted here.

\subsection{Reformulation}

\textbf{(1) Deterministic Reformulation of Expectation Cost Function:} Assume that the generator cost function admits a polynomial form, i.e., $G(x) = \sum_{i=0}^{n} C_{i} x^{i}$, then the objective function~\eqref{obj} can be reformulated as~\eqref{Ed}-\eqref{Ep}. If the cost function is quadratic (n=2), we can directly derive the explicit expression of $\mathbb{E}(\bm{d}_{t})^{k}$ for $k=1,2$. But if the cost function is super-quadratic (\( n > 2 \)), we should assume a specific distribution (e.g., Gaussian distribution) for net load forecast error, and we can derive the explicit expression of $\mathbb{E}(\bm{d}_{t})^{k}$ as~\eqref{Ednormal}. Note that, due to the symmetry of the Gaussian distribution, all odd terms in~\eqref{Ednormal} are zero.
\begin{subequations}
    \begin{align}
&\mathbb{E}\big({{G}}\left( {{g}_{t}}+\varphi_{t}\bm{{d}_{t}} \right)\big)=\sum_{i=0}^n C_i \sum_{k=0}^i\binom{i}{k} g_t^{i-k} \varphi_t^k \mathbb{E}(\bm d_{t})^{k}\label{Ed}\\
        &\mathbb{E}\big( M(p_{t}+\psi_{t}\bm{{d}_{t}})\big)=M(p_{t}+\psi_{t}{\mu}_{t})\label{Ep}\\
       & \mathbb{E}(\bm d_{t})^{k}=\sum_{j=1}^{k}\binom{k}{j} \mu_{t}^{k-j} \sigma_{t}^j(j-1)!!\label{Ednormal}\end{align}\label{expectation}
\end{subequations}

\textbf{(2) Approximation of Joint Chance-Constraints:} By leveraging Bonferroni Approximation~\cite{nemirovski2007convex}, joint chance-constraints~\eqref{rgulbound}~and~\eqref{rsulbound}~can be decomposed into~\eqref{ba}-\eqref{baa}. Hereby, ${{a}_{i}}{(\bm{x})}\text{, }{{b}_{i}}(\textbf{x})$ are affine functions of decisions $\bm{x}$. $\bm{\xi}$ defines uncertainty. \textit{N} is the number of chance-constraints considered jointly.  However, the Bonferroni Approximation does not specify the value of $\epsilon_{i}$. One feasible solution is to use the equal risk allocation, i.e., $\epsilon_{i}=\epsilon/N$. 
Another possible way is to set $\epsilon_{i}$ based on market practices~\cite{wu2014chance}. 
\begin{subequations}
\begin{align}
   & \mathbb{P}\left( {{a}_{i}}{{(\bm{x})}^{\text{T}}}\bm{\xi}(\bm{x}) \le {{b}_{i}}(\bm{x}) \right) \ge 1-\epsilon_{i}\text{,}\quad i=1\text{,}2\text{,}\cdots\text{,}N \label{ba}\\
   & \sum\nolimits_{i=1}^{N}{{{\epsilon }_{i}}}\le \epsilon\text{,}\quad {\epsilon }_{i}\ge 0\text{,}\quad i=1\text{,}2\text{,}\cdots\text{,}N \label{baa}
\end{align}
\end{subequations}

\textbf{(3) Deterministic Reformulation of Chance-Constraints:} Given that the net load forecast error $\bm{{d}_{t}}$ follows a moment-based distribution, i.e., $\bm{{d}_{t}}\sim f\left(\mu_{t}, \sigma_{t}^2\right)$, where~$\mu_{t}$~\text{,}~$\sigma_{t}\text{,}$ and $f$ are the mean, standard deviation and probability density function, respectively. The chance-constraints~\eqref{rgulbound},~\eqref{pcbound}-\eqref{rsulbound} admit a deterministic reformulation in~\eqref{dr}. $F^{-1}$ is the inversed cumulative distribution function (CDF) of the distribution. Quantiles $\widehat{d}{t}$ and $\widetilde{d}{t}$ are given by: $\widehat{d}{t}=\mu_{t}- F^{-1} (1-\epsilon_{i})\sigma_{t}$ and $\widetilde{d}{t}=\mu_{t}+ F^{-1} (1-\epsilon_{i})\sigma_{t}$.
\begin{subequations}
    \begin{align}
&  \underline{G}\le g_{t}+\varphi_{t}\widehat{d}_{t}\text{, }g_{t}+\varphi_{t}\widetilde{d}_{t}\le \overline{G}\label{dr1}\\
& b_{t}-\psi_{t}\widehat{d}_{t}\le \overline{P}\text{, }p_{t}+\psi_{t}\widetilde{d}_{t}\le \overline{P}\label{dr2}\\
& (\psi_{t}\widetilde{d}_{t}+p_{t})/{\eta}\le e_{t}\text{, }e_{t}\le \overline{E}-(b_{t}-\psi_{t}\widehat{d}_{t})\eta\label{dr3}
\end{align}\label{dr}
\vspace{-0.5cm}
\end{subequations}  
\begin{remark}
The inverse CDF function can be obtained based on the following uncertainty realizations:
\begin{enumerate}
\item For uncertainty with assumed distribution, e.g., Gaussian distribution, then we have~$F^{-1}(1-\epsilon)=\Phi^{-1}(1-\epsilon)$.

\item For uncertainty with ambitious and incomplete information, then we can obtain the robust approximation of inversed CDF function from generalizations of the Cantelli’s inequality~\cite{qi2023}. Please refer to Table~\ref{approximation}.

\item For uncertainty with discrete historical observations, we can model the uncertainty by three-parameters \textit{Versatile Distribution} proposed in~\cite{versatile} and learn the parameters by Maximum Likelihood Estimation, then we have $F^{-1}(1-\epsilon\mid a\text{,}b\text{,}c)=c-\ln \left((1-\epsilon)^{-1 / b}-1\right)/{a}$.
\end{enumerate}
\end{remark}

\begin{table}[!ht]
  \centering
  \caption{Robust Approximation of Normalized Inverse Cumulative Distribution with Incomplete Information}
  \setlength{\tabcolsep}{0.2mm}{
      \begin{tabular}{c c c}
    \toprule
    Type \& Shape & ${{F}^{-1}}(1-\epsilon)$ & $\epsilon$ \\
    \midrule
    No Distribution Assumption (NA) & $\sqrt{(1-\epsilon )/\epsilon }$ & $0<\epsilon \le 1$ \\\specialrule{0em}{0.5em}{0em}
    \multirow{2}[0]{*}{ Symmetric Distribution (S)} & $\sqrt{1/2\epsilon }$ & $0<\epsilon \le 1/2$ \\\specialrule{0em}{0.5em}{0em}
          & 0     & $1/2<\epsilon \le 1$ \\\specialrule{0em}{0.5em}{0em}
    \multirow{2}[0]{*}{ Unimodal Distribution (U)} & $\sqrt{(4-9\epsilon )/9\epsilon }$ & $0<\epsilon \le 1/6$ \\\specialrule{0em}{0.5em}{0em}
          & $\sqrt{(3-3\epsilon )/(1+3\epsilon )}$ & $1/6<\epsilon \le 1$ \\\specialrule{0em}{0.5em}{0em}
    \multirow{3}[0]{*}{{ Symmetric \& Unimodal Distribution (SU)}} & $\sqrt{2/9\epsilon }$ & $0<\epsilon \le 1/6$ \\\specialrule{0em}{0.5em}{0em}
          & $\sqrt{3}(1-2\epsilon )$ & $1/6<\epsilon \le 1/2$ \\\specialrule{0em}{0.5em}{0em}
          & 0     & $1/2<\epsilon \le 1$ \\\specialrule{0em}{0.5em}{0em}
    \bottomrule
    \end{tabular}%
    }\vspace{-0.5cm}
  \label{approximation}%
\end{table}%

\subsection{Robust Competitive Equilibrium}

\begin{prop}\label{p1}
   Given the relaxation of~\eqref{st}~\cite{nazir2021guaranteeing} or the charge and discharge states, let~$\{g_t^{*}\text{, }p_{t}^{*}\text{, }b_{t}^{*}\text{, }e_{t}^*\text{, }\varphi_{t}^{*}\text{, }\psi_{t}^{*}\}$ be an optimal solution of \eqref{MCCED}~and~$\{\lambda_t^{*}\text{, }\theta_t^{*}\text{, }\pi_t^{*}\}$~be dual variables of constraints~\eqref{pb}-\eqref{rgt}. Then, $\{g_t^{*}\text{, }p_{t}^{*}\text{, }b_{t}^{*}\text{, }e_{t}^*\text{, }\varphi_{t}^{*}\text{,  }\psi_{t}^{*}\text{, }\lambda_t^{*}\text{, }\theta_t^{*}\text{, }\pi_t^{*}\}$
   constitutes a robust
competitive equilibrium, i.e.,  
\begin{enumerate}
    \item The system operator clear the market at: ${g}_{t}-p_{t}+b_{t}=D_{t}$, ${e}_{t+1}-{e}_{t}=-p_{t}/{\eta }+b_{t}{{\eta }}$, and $\varphi_{t}+\psi_{t}=1$.
    \item Each producer maximizes its profit under the payment: $\lambda_t^{*}g_t^{*}$~for the energy produced by the conventional generator, $\theta_t^{*}(p_{t}^{*}/{\eta }-b_{t}^{*}{\eta })$~for the energy produced by energy storage, and $\pi_t^{*}(\varphi_{t}+\psi_{t})$~for the reserve provided by conventional generator and energy storage.
\end{enumerate}
\end{prop}

\begin{proof}
The proof is trivial. Since problem~\eqref{MCCED} is a linear programming problem after relaxation and deterministic reformulation, the \textit{Simplex} guarantees that the optimum is attained at the vertices (i.e., bounds of inequality), and \textit{Duality Theory} guarantees a robust equilibrium.
\end{proof}

We provide the Lagrange function and KKT conditions of problem~\eqref{MCCED} in Appendix~\ref{Lagrange}.

\section{Main Results}\label{MR}

This section presents a theoretical pricing analysis to demonstrate the SoC monotonicity and uncertainty monotonicity of storage opportunity price. We then prove that the storage market power can be regulated and the opportunity price is bounded by energy and reserve prices.

\subsection{SoC Monotonicity of Storage Opportunity Price}
We first show that the opportunity cost function of the storage SoC is convex, i.e., the opportunity price (derivative of the opportunity cost) monotonically decreases with SoC. Hence, it can be efficiently incorporated into economic dispatch using piecewise linear approximations similar to generator cost curves. To demonstrate this, the following proposition proves that the price of storage opportunity monotonically decreases with the storage SoC.
\begin{prop}\label{prop2} \textbf{Convex opportunity cost function.}
   Given a monotonically increasing and quadratic/super-quadratic cost function~$G$, if the storage power is not binding with the power bounds, then we have ${\partial {{\theta }_{t}}}/{\partial {{e}_{t}}}\le 0$.  If the storage power is binding with the power bounds, then we have ${\partial {\sup({\theta }_{t}})}/{\partial {{e}_{t}}}\le 0$, ${\partial {\inf({\theta }_{t}})}/{\partial {{e}_{t}}}\le 0$. 
\end{prop}

This result fits the diminishing storage value; hence, the marginal value decreases with storage SoC levels. Note that, as shown in~\eqref{SoC}, discharge power reduces SoC while charge power increases SoC. Hence, Proposition~\ref{prop2} shows that given an initial SoC, the discharging price will increase with higher discharge power, while the willingness to charge decreases with more charging power. This aligns with current convex generator cost curves and flexible demand bids. We defer the complete proof to Appendix~\ref{appendix1}.

Proposition~\ref{prop2} also aligns with prior studies that show the storage opportunity function is convex, but were derived using a price-taker profit maximization framework~\cite{xu2020operational, zheng2023energy}. In contrast, we derive Proposition~\ref{prop2} using a social welfare maximization framework. These results indicate that a convex SoC cost function can serve both system operator dispatch (social welfare maximization) and the bids submitted by storage participants (profit maximization).


\begin{corollary}\label{cl1}\textbf{Eliminated bounds gap in ideal storage.}
   Given a monotonically increasing and quadratic cost function~$G$, if $\eta=1$, then ${\partial {{\theta }_{t}}}/{\partial {{e}_{t}}}\le 0$ always holds true.  
\end{corollary}

\begin{proof}
   The proof is trival that the upper and lower bounds of $\theta$ are equal when $\eta=1$, hence ${\partial {{\theta }_{t}}}/{\partial {{e}_{t}}}$= ${\partial {\sup({\theta }_{t}})}/{\partial {{e}_{t}}}={\partial {\inf({\theta }_{t}})}/{\partial {{e}_{t}}}\le 0$.
\end{proof}
Corollary~\ref{cl1} demonstrates that the gap between the lower and upper bounds of the opportunity cost function's subgradient will be eliminated for ``ideal'' storage. The rationale is that without energy conversion losses, the increased discharge price rate with SoC should equal the decreased charging price rate with SoC. Moreover, the first two cases in the proof verify that the charging and discharging prices should be the same when overlooking the marginal degradation cost and energy conversion losses.

\subsection{Uncertainty Monotonicity of Storage Opportunity Price}
We further demonstrate that the storage opportunity price increases with system uncertainty. This is a significant difference compared to conventional thermal generators, whose cost curve should always be based on fuel cost curves regardless of the uncertainty exposed.
\begin{prop}\label{prop3}\textbf{Uncertainty-aware opportunity cost function.}
   Given a super-quadratic function $G$, then we have $\partial {{\theta }_{t}}/\partial {{\sigma }_{t}}> 0$. While, given a quadratic function $G$, then we have $\partial {{\theta }_{t}}/\partial {{\sigma }_{t}}=0$. 
\end{prop}
\begin{proof}
Take Case 1~\eqref{case1} in proof of Proposition~\ref{prop2} as an example, we assume that forecast error follows a normal distribution, then we have~\eqref{dsigma}. If $G$ follows a super-quadratic function ($n>2$), there are cross-terms of $\sigma_{t}$ and $g_{t}$ in the expectation cost function, hence the derivative function in~\eqref{dsigma} should be positive. If $G$ follows a quadratic function ($n=2$), there is no cross-terms of $\sigma_{t}$ and $g_{t}$ in the expectation cost function, hence the derivative function in~\eqref{dsigma} should be zero. And we provide results for  cases $n=2\text{,}3\text{,}4$ in~\eqref{dn2}-\eqref{dn4}.
\begin{subequations}
    \begin{align}
&\dfrac{\partial{\theta }_{t}}{\partial \sigma_{{t}}}=\dfrac{1}{\eta}\sum_{i=0}^n C_i \sum_{k=0}^i\binom{i}{k} g_t^{i-k-1} \varphi_t^k (i-k)\times\label{dsigma}\\ 
&\hspace{0.6cm}\sum\limits_{\text 0 \leq j \leq k}^{j\text{ }even}\binom{k}{j} \mu_{t}^{k-j} {\sigma}_{t}^{j-1} j (j-1)!!\notag\\
   & {{\theta }_{t}}=\dfrac{1}{\eta}\big( C_1+C_2\left(2g_t+2 \varphi_t \mu_{t}\right)\big)\text{, }\dfrac{\partial{\theta }_{t}}{\partial \sigma_{{t}}}=0\text{, }n=2\label{dn2}\\
     &  {{\theta }_{t}}=\dfrac{1}{\eta}\big( C_1+2 C_2 g_t+2 C_2 \varphi_t \mu_t+3 C_3 g_t^2+6 C_3 g_t \varphi_t \mu_t\label{dn3}\\
     &+3 C_3 \varphi_t^2\left(\mu_t^2+\sigma_t^2\right)\big)\text{, }\dfrac{\partial{\theta }_{t}}{\partial \sigma_{{t}}}=6 C_3\varphi_t^2\sigma_t/\eta>0\text{, }n=3 \nonumber\\
     & {{\theta }_{t}}=\dfrac{1}{\eta}\big(C_1 + 2 C_2 g_t + 2 C_2 \varphi_t \mu_t + 3 C_3 g_t^2 + 6 C_3 g_t \varphi_t \mu_t \label{dn4}\\
     &+ 3 C_3 \varphi_t^2 (\mu_t^2 + \sigma_t^2) + 4 C_4 g_t^3 + 12 C_4 g_t^2 \varphi_t \mu_t \nonumber\\
     &+ 12 C_4 g_t \varphi_t^2 (\mu_t^2 + \sigma_t^2)+ 4 C_4 \varphi_t^3 (\mu_t^3 + 3 \mu_t \sigma_t^2)\nonumber\\
     &\dfrac{\partial{\theta }_{t}}{\partial \sigma_{{t}}}=\sigma_t\left(6 C_3 \varphi_t^2+24 C_4 g_t \varphi_t^2+24 C_4 \varphi_t^3 \mu_t\right)/\eta>0\text{, }n=4\nonumber
    \end{align}
\end{subequations}
Hence, we finished the proof.
\end{proof}
Note that the uncertainties lie not only in the net load but also in the look-ahead window of the dispatch. Proposition~\ref{prop3} demonstrates that with a super-quadratic function of \( G \), the storage opportunity cost will increase with higher non-anticipativity either from higher penetration of uncertain resources (e.g., renewables, electric vehicle) or
a longer look-ahead window. However, given a quadratic cost function, the opportunity cost remains stable. Proposition~\ref{prop3} also aligns with prior studies~\cite{xu2020operational, qin2024economic} that have shown the storage opportunity value scales with price uncertainty in price-taker profit-maximization objectives. Yet, we consider net load uncertainty from a social welfare maximization perspective.

\begin{prop}\label{prop4} \textbf{Mapping uncertainty with nonlinear opportunity cost function.}
Given a  monotonically increasing and super-quadratic function $G$, then we have $\mathbb{E}({{\theta }_{t}}(\bm{{d}_{t}}))>{{\theta }_{t}}(\mathbb{E}(\bm{{d}_{t}}))$.   
\end{prop} 
\begin{proof}
    Take Case 1~\eqref{case1} in proof of Proposition~\ref{prop2} as an example. The relationship between~\eqref{er1}~and~\eqref{er2} depends on the convexity of $\partial G$. If $G$ is a super-quadratic function, then $\partial G$ will be a convex (but not linear) function, hence, according to Jensen's inequality~\cite{dekking2005modern}, $\mathbb{E}({{\theta }_{t}}(\bm{{d}_{t}}))>{{\theta }_{t}}(\mathbb{E}(\bm{{d}_{t}}))$. The difference between the two sides of the inequality, is called the Jensen gap. 
\begin{subequations}
\begin{align}
   & \mathbb{E}({{\theta }_{t}}(\bm{ d_{{t}}}))=\dfrac{1}{{{\eta }}}\dfrac{\partial \mathbb{E} G( D_{t}+\varphi_{t}\bm{{d}_{t}}+(e_{t}-e_{t-1})/{\eta }))}{\partial ( D_{t}+\varphi_{t}\bm{{d}_{t}}+(e_{t}-e_{t-1})/{\eta })} \label{er1}\\
   & {{\theta }_{t}}(\mathbb{E}(\bm{d_{{t}}}))=\dfrac{1}{{{\eta }}}\dfrac{\partial G(D_{{t}}+\varphi_{t}\mathbb{E}( \bm{d_{{t}}})+(e_{t}-e_{t-1})/{\eta })}{\partial (D_{{t}}+\varphi_{t}\mathbb{E}( \bm{d_{{t}}})+(e_{t}-e_{t-1})/{\eta })}\label{er2}
\end{align}
\end{subequations}
Hence, we have finished the proof.
\end{proof}
\begin{corollary}\label{cl2}
    \textbf{Mapping uncertainty with linear opportunity cost function.} Given a monotonically increasing and quadratic function $G$, then we have $\mathbb{E}({{\theta }_{t}}(\bm{{d}_{t}}))={{\theta }_{t}}(\mathbb{E}(\bm{{d}_{t}}))$.
\end{corollary}
\begin{proof}
    This corollary is trivial to show. If $G$ is a monotonically increasing and quadratic function, then $\partial G$ will be a linear function, hence, according to Jensen's inequality, $\mathbb{E}({{\theta }_{t}}(\bm{{d}_{t}}))={{\theta }_{t}}(\mathbb{E}(\bm{{d}_{t}}))$. 
\end{proof}
Proposition~\ref{prop4} and Corollary~\ref{cl2} are the extreme cases of Proposition~\ref{prop3}, the left term of the inequality represents the case with uncertainty, while the right term represents the case without uncertainty. These results can also be explained by mapping uncertainty with the opportunity cost function: the super-quadratic case will introduce skewness and increase the mean value, while the quadratic case is a linear mapping.


\subsection{Regulated Market Power of Energy Storage}\label{OER}

Finally, we show that the storage opportunity price not only depends on future uncertainties but is also constrained by future energy and reserve prices. 
In practice, the system operator usually imposes price floors and ceilings. Hence, it is beneficial to use the proposed pricing mechanism as it can bind the opportunity price and regulate storage market power.

\begin{theorem}\label{tm} \textbf{Constrained opportunity price.}
The opportunity price of storage has a linear coupling relationship with the energy price and reserve cost. 
\begin{enumerate}
\item Charging: $\theta_{t-1}=\dfrac{\eta}{\widetilde{d}_{t}}\big(\theta_{t}\eta\widehat{d}_{t}+\lambda_{t}(\dfrac{\widetilde{d}_{t}}{\eta^2}-\widehat{d}_{t})+\pi_{t}-M\mu_{t} \big)$

\item Discharging: $\theta_{t-1}=\dfrac{1}{\eta\widehat{d}_{t}}\big(\dfrac{\theta_{t}\widetilde{d}_{t}}{\eta}+\lambda_{t}(\eta^2\widehat{d}_{t}-\widetilde{d}_{t})+\pi_{t}+M(\widetilde{d}_{t}-\eta^2\widehat{d}_{t}-\mu_{t}) \big)$ 

\item Idle: $\dfrac{1}{\eta\widehat{d}_{t}}\big(\dfrac{\theta_{t}\widetilde{d}_{t}}{\eta}+\lambda_{t}(\eta^2\widehat{d}_{t}-\widetilde{d}_{t})+\pi_{t}+M(\widetilde{d}_{t}-\eta^2\widehat{d}_{t}-\mu_{t}) \big)$$\leq\theta_{t-1}\leq$$\dfrac{\eta}{\widetilde{d}_{t}}\big(\theta_{t}\eta\widehat{d}_{t}+\lambda_{t}(\dfrac{\widetilde{d}_{t}}{\eta^2}-\widehat{d}_{t})+\pi_{t}-M\mu_{t}  \big)$
\end{enumerate}
\end{theorem}

Theorem~\ref{tm} states that if storage jointly participates in the energy and reserve markets, the opportunity price will be constrained by the energy and reserve prices as well as the opportunity price from the past time. Moreover, this relationship varies depending on the charging and discharging states. Net load uncertainty will also impact the opportunity price. We defer the complete proof to Appendix~\ref{appendix2}.

\begin{corollary}\label{cl3} 
\textbf{Bounded opportunity price.}
Assume energy and reserve price have the bounds: $\lambda_{t}\in[\underline{\lambda}\text{, }\overline{\lambda}]$, $\pi_{t}\in[\underline{\pi}\text{, }\overline{\pi}]$, then the storage opportunity price are bounded: 

\begin{enumerate}
\item Charging Price: $\dfrac{\eta}{\widetilde{d}_{t}}\big(\underline{\lambda}(\dfrac{\widetilde{d}_{t}}{\eta^2}-\widehat{d}_{t})+\underline{\pi}-M\mu_{t}\big)\le\overline{\theta}_{t}\le\dfrac{-1}{\eta\widehat{d}_{t}}\big(\overline{\lambda}(\dfrac{\widetilde{d}_{t}}{\eta^2}-\widehat{d}_{t})+\overline{\pi}-M\mu_{t}\big)$

\item Discharging Price: $\dfrac{1}{\eta\widehat{d}_{t}}\big(\underline{\lambda}(\eta^2\widehat{d}_{t}-\widetilde{d}_{t})+\overline{\pi}+M(\widetilde{d}_{t}-\eta^2\widehat{d}_{t}-\mu_{t})\big)\le\overline{\theta}_{t}\le\dfrac{-\eta}{\widetilde{d}_{t}}\big(\overline{\lambda}(\eta^2\widehat{d}_{t}-\widetilde{d}_{t})+\underline{\pi}+M(\widetilde{d}_{t}-\eta^2\widehat{d}_{t}-\mu_{t})\big)$ 
\end{enumerate}
\end{corollary}
\begin{proof}
This corollary is trivial to demonstrate. Following Theorem~\ref{tm}, we can derive the opportunity price bounds after setting the energy and reserve prices as the bound values.
\end{proof}
Corollary~\ref{cl3} states that the opportunity price can be directly limited by the system operator. This aligns with previous work~\cite{qin2024economic} that has shown the upper bounds of storage bids with peak energy price limitation. However, in market practice, the opportunity price is unlikely to reach these theoretical bounds, as the extreme cases of energy and reserve prices typically do not occur simultaneously.

If the storage only participates in the energy market, we cannot derive these constrained and bounded relationship. However, the previously stated propositions and corollaries still hold true.  We will demonstrate the market power of storage under different market participations in subsequent case studies.

\section{Numerical Case Study}\label{Case Study}

\subsection{Set-Up}

We demonstrate the effectiveness of the proposed theoretical results with the following case studies based on the 8-zone ISO-NE test
system~\cite{krishnamurthy20158}~with an average load of 13 GW. The test system includes 76 thermal generators with a total installed capacity of 23.1 GW. 
Renewables and energy storage are configured based on certain percentage of load capacity. The forecast error of renewables and load are generated from Elia~\cite{Elia-data}. 
We set the renewables capacity to be 30\%, storage capacity and duration to be 20\% and 4-hr. Efficiency, marginal cost, and initial SoC of storage are set to be 95\%, \$20/MWh, and 50\%, unless otherwise specified.

The optimization is coded in MatLab and solved by Gurobi 11.0 solver. The programming environment is Intel Core i9-13900HX @ 2.30GHz with RAM 16 GB\footnote{The code and data used in this study are available at: \url{https://github.com/thuqining/Storage_Pricing_for_Social_Welfare_Maximization.git}}.
         
\subsection{Prices versus Risk-Aversion and Uncertainty Realizations}
The first case study shows the dependency of system prices and costs over risk preferences.
We compare the average energy price ($\mathbb E\lambda_{t}$), average storage opportunity price ($\mathbb E\theta_{t}$), total reserve cost ($\sum\pi_{t}$), and system cost of one simulation day under different risk aversion levels ($\epsilon$) and uncertainty realizations in Figure~\ref{price-risk}. Results show that as risk-aversion decreases (increase $\epsilon$), the reserve and system costs decrease. At the same time, both energy and storage prices are declined, which is consistent with the theoretical results from~\eqref{KKT-1}.

Additionally, to quantify the distribution approximation error, we calculate the RMSE between prices and costs under different distributions and those under the baseline empirical results (using quantiles from historical data) and summarize the results in Table~\ref{aperror}. The quantitative results show that the versatile distribution outperforms the other methods in terms of the optimality gap due to the least fitting error. Gaussian distribution is generally good, but the optimality gap increases with larger $\epsilon$ ($\epsilon\geq0.3$). However, using symmetric and unimodal distributions produces extremely conservative results for extremely risk-averse scenarios ($\epsilon \leq 0.05$), leading to a significant gap compared to the baseline empirical results. The above results demonstrate that the proposed pricing mechanism is adaptable to different levels of risk aversion and can generate relatively accurate prices under versatile and Gaussian distribution assumptions for forecast error.
\begin{figure}[!ht]
\vspace{-0.5em}
      \setlength{\abovecaptionskip}{-0.1cm}  
    \setlength{\belowcaptionskip}{-0.1cm} 
  \begin{center}  \includegraphics[width=1\columnwidth]{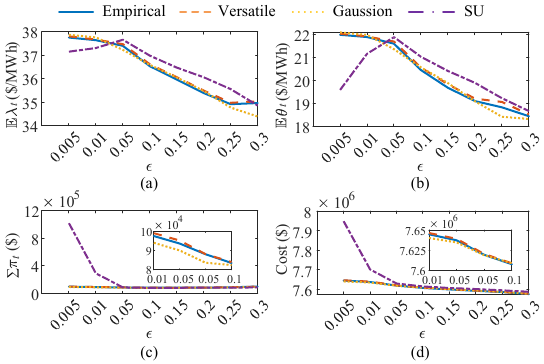}
    \caption{Prices and costs under different risk aversions and different uncertainty realization: (a) energy price, (b) storage opportunity price, (c) reserve cost, and (d) system cost.}\label{price-risk}
  \end{center}
  \vspace{-0.5cm}
\end{figure}

\begin{table}[!ht]
  \centering
  \caption{RMSE between Prices and Costs under Different Distributions and Those under the Baseline Empirical Results}
  \setlength{\tabcolsep}{1.2mm}{
      \begin{tabular}{c c c c c}
       \toprule
    Distributions & Energy Price & Storage Price & Reserve Cost & System Cost \\
    \midrule
    Versatile & 0.08 & 0.12 & 8.99E+02 & 1.28E+03 \\
    Gaussion & 0.23 & 0.20 & 3.44E+03 & 2.20E+03 \\
    SU    & 0.48 & 1.00 & 3.33E+05 & 1.10E+05 \\
    \bottomrule
    \end{tabular}%
    }\vspace{-0.3cm}
  \label{aperror}%
\end{table}%

\subsection{Opportunity Prices versus SoC and Net load Uncertainty}

We fix the risk-aversion level at $\epsilon=0.05$ and further investigate the storage opportunity prices under different SoC and net load uncertainty levels. We use the piecewise-quadratic cost function for conventional generators from the original ISO-NE test system, which is approximately super-quadratic as illustrated in Figure~\ref{costfunction}. By varying the initial SoC and the scale factor of $\sigma$, the results illustrated in Figure~\ref{figprop2} demonstrate that the storage opportunity price monotonically decreases with SoC, which verifies the Proposition~\ref{prop2}. When the initial SoC equals 100\%, the storage has no further charging capability or opportunity value; thus, its opportunity price drops dramatically to 0. 

The opportunity price generally increases with net load uncertainties, which verifies the Proposition~\ref{prop3}. However, given a highly volatile net load (scale factor $\geq 2.1$), the conventional generator will hit the lower bound due to insufficient downward reserve capacity, which will slightly decrease the energy price due to positive value of $\underline{\nu}_{t}$ in~\eqref{lg2}. And the proportional relationship between energy price and storage opportunity price from~\eqref{dc2} or~\eqref{dd2} explains the slight decrease in storage opportunity price in this case.
\begin{figure}[!ht]
\vspace{-0.5cm}
      \setlength{\abovecaptionskip}{-0.1cm}  
    \setlength{\belowcaptionskip}{-0.1cm} 
  \begin{center}  \includegraphics[width=0.82\columnwidth]{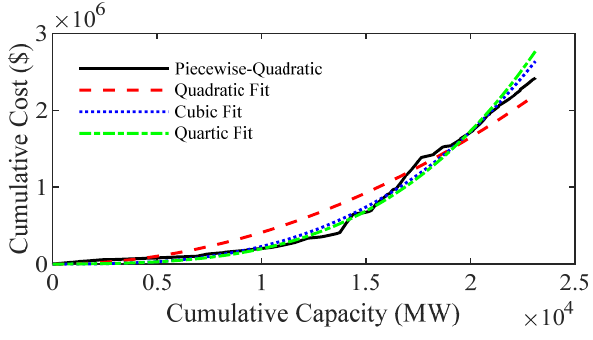}
    \caption{Cumulative cost curve of 76 conventional generators in the ISO-NE test system.}\label{costfunction}
  \end{center}
  \vspace{-0.5cm}
\end{figure}
\begin{figure}[!ht]
\vspace{-0.5em}
      \setlength{\abovecaptionskip}{-0.1cm}  
    \setlength{\belowcaptionskip}{-0.1cm} 
  \begin{center}  \includegraphics[width=0.82\columnwidth]{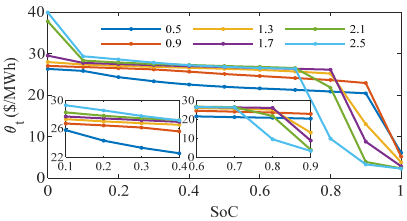}
    \caption{Opportunity prices under different initial SoC and uncertainty levels.}\label{figprop2}
  \end{center}
  \vspace{-0.3cm}
\end{figure}

We then test the sensitivity of opportunity price over uncertainty under different cost functions of the generator and summarize the average performance in Figure~\ref{figureprop4}. The cost function from the original test system is piecewise-quadratic (blue line). We remove the quadratic term to create a piecewise-linear cost function (red line), and then generate the quadratic function through fitting (yellow line). It is observed that the opportunity price increases with the net load uncertainty under the piecewise-quadratic and piecewise-linear functions. In contrast, under the fitted quadratic cost functions, the opportunity price remains fixed at \$19.4.  Moreover, it can be observed that as $\sigma$ continues to approach zero, under the piecewise-quadratic (super-quadratic) function, $\mathbb{E}({{\theta }_{t}}(\bm{{d}_{t}}))$ is always greater than ${{\theta }_{t}}(\mathbb{E}(\bm{{d}_{t}}))$. In contrast, under the quadratic function, $\mathbb{E}({{\theta }_{t}}(\bm{{d}_{t}}))$ equals ${{\theta }_{t}}(\mathbb{E}(\bm{{d}_{t}}))$. The above simulation results verify the Proposition~\ref{prop4} and Corollary~\ref{cl2}.
\begin{figure}[!ht]
\vspace{-0.5em}
      \setlength{\abovecaptionskip}{-0.1cm}  
    \setlength{\belowcaptionskip}{-0.1cm} 
  \begin{center}  \includegraphics[width=0.82\columnwidth]{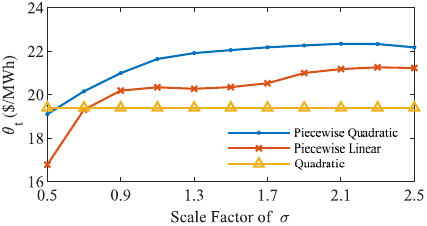}
    \caption{Expected opportunity prices under different cost functions and uncertainty levels.}\label{figureprop4}
  \end{center}
  \vspace{-0.3cm}
\end{figure}

\subsection{Regulated Market Power of Energy Storage}

In this case, we fix the confidence level at $\epsilon=0.05$ and set the scale factor of $\sigma$ to 1. Figure~\ref{figuretheorem} illustrates the opportunity prices simulated from dual variables (black solid lines) and calculated from the theorem (blue and red dashed lines) over four representative days. We observe that for charging/discharging states, the simulated opportunity prices equal the theoretical result in~\eqref{reserver1}-\eqref{reserver2}. For idle states, the simulated opportunity prices fall within the theoretical bounds as~\eqref{reserver5}. These results verify Theorem~\ref{tm}, showing that the storage opportunity price is linearly constrained by energy price and reserve price, and this relationship varies with different charging and discharging states. Furthermore, the charging price and discharging price are bounded within [0, 30.05]\$/MWh and [0, 47.63]\$/MWh for the whole-year simulation, which verifies Corollary~\ref{cl3}. The bounded opportunity prices demonstrate that the market power of energy storage can be regulated by the system operator using the proposed pricing mechanism.
\begin{figure}[!ht]
\vspace{-0.5em}
      \setlength{\abovecaptionskip}{-0.1cm}  
    \setlength{\belowcaptionskip}{-0.1cm} 
  \begin{center}  \includegraphics[width=1\columnwidth]{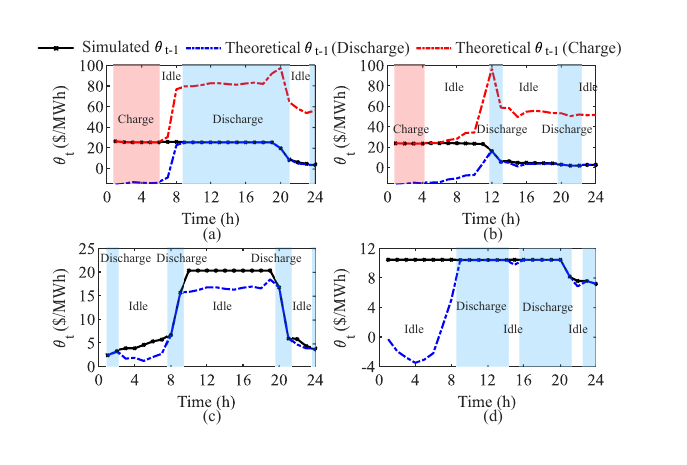}
    \caption{Simulated and theoretical opportunity prices over four simulation days.}\label{figuretheorem}
  \end{center}
  \vspace{-0.5cm}
\end{figure}

Furthermore, Figure~\ref{figuremarket} illustrates the market power of energy storage under different market participations. Compared with the baseline results (non-participation), the increase in storage capacity results in a significant reduction in storage opportunity prices, energy prices, and system costs. Compared to participating only in the energy market, prices and costs witness a faster and greater reduction when storage participates in both the energy and reserve markets. This indicates that storage exercises greater market power when additionally participating in the reserve market. Additionally, it can be observed that in the energy \& reserve provision case, storage has the greatest market power at 30\%-50\% capacity, whereas in the energy provision scenario, storage has the greatest market power at 80\%-90\% capacity. 
\begin{figure}[!ht]
\vspace{-0.5em}
      \setlength{\abovecaptionskip}{-0.1cm}  
    \setlength{\belowcaptionskip}{-0.1cm} 
  \begin{center}  \includegraphics[width=0.98\columnwidth]{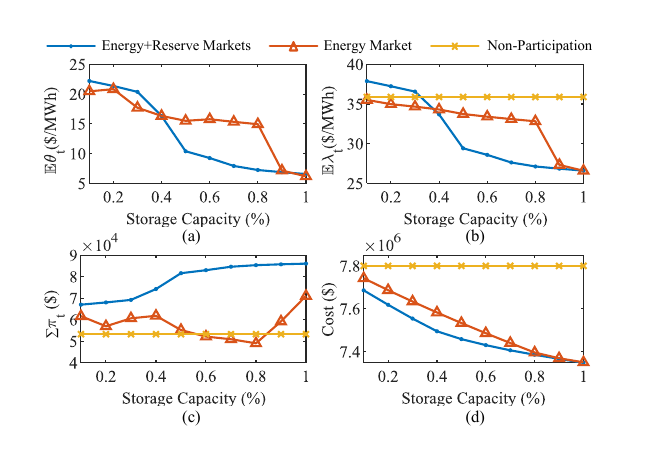}
    \caption{Prices and costs under different market participations of energy storage: (a) energy price, (b) storage opportunity price, (c) reserve cost, and (d) system cost.}\label{figuremarket}
  \end{center}
  \vspace{-0.3cm}
\end{figure}

\subsection{Benefit of the Proposed Pricing Framework}

We further compare the economic performance of proposed social welfare maximization framework with storage profit maximization framework in~\cite{qin2024economic}. The formulation of the storage profit maximization can refer to Appendix~\ref{appendix3}. In order to ensure comparability between the two pricing mechanisms, we exclude the reserve from storage and simulate the real-time energy price based on the multi-period economic dispatch with Monte Carlo simulation of net load. Additionally, we retire 20\% of conventional generators to increase price variations.

We compare the economic performance of one representative day in Figure~\ref{figurecompare}. It shows that with an increase in net load forecast error, the storage profit, conventional generation cost, system cost, and consumer electricity payment significantly increase under both pricing mechanisms. Furthermore, compared to the storage profit maximization mechanism, the storage profit, conventional generation cost, system cost, and electricity payment drop by an average of 25.4\%, 3.9\%, 3.9\%, and 17.4\%, respectively. And the decrease in storage profit only accounts for 0.49\% of the electricity payment. The above results verify that while the proposed pricing mechanism sacrifices a portion of storage profit, it significantly enhances social welfare by reducing system costs, conventional generator production, and consumer payments. 
\begin{figure}[!ht]
\vspace{-0.5em}
      \setlength{\abovecaptionskip}{-0.1cm}  
    \setlength{\belowcaptionskip}{-0.1cm} 
  \begin{center}  \includegraphics[width=0.98\columnwidth]{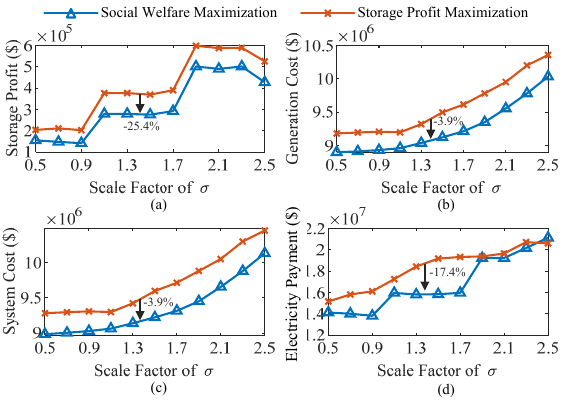}
    \caption{Economic performance under different pricing mechanisms and uncertainty levels: (a) storage profit, (b) conventional generation cost, (c) system cost, and (d) electricity payment.}\label{figurecompare}
  \end{center}
  \vspace{-0.3cm}
\end{figure}

\subsection{Sensitivity Analysis}
In this subsection, we further investigate the performance and benefits of the proposed pricing mechanism as the renewables and storage capacity scale up.

\textbf{(1) Storage Capacity.} Table~\ref{storagecapacity} summarizes the performance of the proposed pricing mechanism under different storage configurations. As storage capacity or duration increases, storage profit increases while all the other costs decrease. Compared to the profit maximization mechanism, increasing storage capacity slightly increases the storage profit reduction rate but significantly enhances the social welfare increase rate. However, increasing storage duration slightly reduces the social welfare increase rate. Specifically, storage profit is reduced by an average of 30\%, while electricity payment decreases by an average of 20\%. Moreover, the decrease in storage profit only accounts for around 8\% of the electricity payment savings. Therefore, in market practice, some of the electricity payment savings can be used to compensate storage providers who choose to take the default opportunity price. The above results verify the superior social welfare performance of the proposed pricing mechanism as storage integration scales up.
\begin{table}[!ht]
  \centering
  \caption{Comparison of Economic Performance under Different Storage Configurations and
   Price Mechanisms}
  \setlength{\tabcolsep}{0.5mm}{
    \begin{tabular}{cccccc}
    \toprule
    \makecell{Storage\\ Duration} & \makecell{Storage\\ Capacity} & \makecell{Storage \\ Profit ($10^5$\$)} & \makecell{Generation \\  Cost ($10^6$\$)} & \makecell{System \\ Cost ($10^6$\$)} & \makecell{Electricity \\ Payment ($10^7$\$)} \\
    \midrule
    \multirow{3}[1]{*}{4-hr} & 20\%  & 1.41(-40\%) & 8.95(-2.5\%) & 9.05(-2.5\%) & 1.38(-18\%) \\
          & 40\%  & 3.37(-15\%) & 8.69(-9.4\%) & 8.81(-10\%)& 1.35(-26\%) \\
          & 60\%  & 5.12(-23\%) & 8.43(-13\%) & 8.59(-14\%) & 1.33(-27\%) \\
    \multirow{3}[0]{*}{8-hr} & 20\%  & 3.37(-26\%) & 8.69(-2.6\%) & 8.81(-3.2\%) & 1.35(-12\%) \\
          & 40\%  & 6.78(-23\%) & 8.17(-12\%) & 8.37(-13\%) & 1.31(-18\%) \\
          & 60\%  & 6.84(-31\%) & 7.71(-10\%) & 8.00(-10\%) & 1.11(-28\%) \\
    \multirow{3}[1]{*}{12-hr} & 20\%  & 5.12(-23\%)  & 8.43(-2.4\%)  & 8.59(-2.6\%)  & 1.33(-9.1\%)  \\
          & 40\%  & 6.87(-30\%) & 7.70(-6.9\%) & 8.00(-6.7\%) & 1.11(-22\%) \\
          & 60\%  & 9.08(-33\%) & 7.09(-11\%) & 7.53(-10\%) & 1.05(-25\%) \\
    \bottomrule
    \end{tabular}%
    }
  \label{storagecapacity}%
\end{table}%

\textbf{(2) Renewables Penetration.} It is observed from Table~\ref{renewablecapacity} that increasing renewable penetration leads to a decrease in all profit and costs, with storage profit and electricity payment being the most affected. This is because increasing renewables penetration reduces the scarcity of energy, which in turn lowers energy and storage prices. Compared to the profit maximization mechanism, increasing renewable capacity slightly increases the storage profit reduction rate but slightly enhances the social welfare increase rate. Specifically, storage profit is reduced by an average of 33\%, while electricity payment decreases by an average of 20\%. Moreover, the decrease in storage profit only accounts for around 2\% of the electricity payment savings. The above results verify the superior social welfare performance of the proposed pricing mechanism as renewable integration scales up.
\begin{table}[!ht]
  \centering
  \caption{Comparison of Economic Performance under Different Renewable Penetration and
   Price Mechanisms}
  \setlength{\tabcolsep}{1.3mm}{
    \begin{tabular}{ccccc}
    \toprule
\makecell{Renewable\\ Capacity} & \makecell{Storage \\ Profit ($10^5$\$)} & \makecell{Generation \\  Cost ($10^6$\$)} & \makecell{System \\ Cost ($10^6$\$)} & \makecell{Electricity \\ Payment ($10^7$\$)} \\
    \midrule
    10\%  & 1.47(-30\%) & 9.18(-4.2\%) & 9.28(-4.1\%) & 1.47(-22\%) \\
    30\%  & 1.41(-39\%) & 8.95(-2.5\%) & 9.05(-2.5\%) & 1.38(-18\%) \\
    50\%  & 1.45(-28\%) & 8.75(-2.9\%) & 8.84(-3.0\%) & 1.31(-21\%) \\
    70\%  & 1.44(-29\%) & 8.58(-3.3\%) & 8.66(-3.4\%) & 1.24(-22\%) \\
    90\%  & 1.36(-41\%) & 8.43(-2.5\%) & 8.52(-2.3\%) & 1.20(-20\%) \\
    \bottomrule
    \end{tabular}%
    }\label{renewablecapacity}%
\end{table}%

\section{Conclusion}\label{Conclusion}
In this paper, we propose a novel framework for energy storage pricing that internalizes the uncertainty of net load and the risk tolerance of the market operator in the price formation process using a chance-constrained approach. We conducted theoretical analysis and simulation validation on key properties of the storage opportunity price, and compared it with profit-maximizing bidding approaches showing the proposed pricing method can effectively reduce system cost and consumer payments.

The proposed pricing model is designed to run separately to determine storage opportunity costs. Hence, it is suitable for practical implementation as its approximation and complexity will not impact the main system dispatch and market clearing process.  System operators can use the proposed approach to generate default bids for storage resources that opt-in for direct system operator management. The generated opportunity price can also serve as a benchmark to monitor storage market power, especially as a base to determine storage bid ceilings to limit the price influence for storage resources.



\bibliographystyle{IEEEtran}
\bibliography{IEEEabrv,CCOP}

\appendix

\subsection{Lagrange Function and KKT Conditions}\label{Lagrange}

The Lagrange function and KKT conditions of problem~\eqref{MCCED} are formulated in~\eqref{KKT}.
\begin{subequations}
\begin{align}
&\hspace{-0.2cm}L= \mathbb{E}\sum\nolimits_{t}\big({{{G}}\left( {{g}_{t}}+\varphi_{t}\bm{{d}_{t}} \right)}+M(p_{t}+\psi_{t}\bm{{d}_{t}})\big)-\underline\alpha_{t}b_{t}-\underline\beta_{t}p_{t}\label{lag1}\\
&\hspace{-0.2cm}+\overline\alpha_{t}(b_{t}-\psi_{t}\widehat{d}_{t}-\overline{P})+\overline\beta_{t}(p_{t}+\psi_{t}\widetilde{d}_{t}-\overline{P})-\pi_{t}(\varphi_{t}+\psi_{t}-1)\nonumber\\
&\hspace{-0.2cm}-{{\lambda }_{t}}( {{g}_{t}}+p_{t}-b_{t}-{{D}_{t}} )+{{\theta }_{t}}\left( {{e}_{t+1}}-{{e}_{t}}+p_{t}/{{\eta }}-b_{t}{{\eta }} \right)\nonumber\\
&\hspace{-0.2cm}+{{\theta }_{t-1}}( {{e}_{t}}-{{e}_{t-1}}+p_{t-1}/{{\eta }}-b_{t-1}{{\eta }} )-\underline\nu_{t}(g_{t}+\varphi_{t}\widehat{d}_{t}-\underline{G})\nonumber\\
&\hspace{-0.2cm}+\overline\nu_{t}(g_{t}+\varphi_{t}\widetilde{d}_{t}-\overline{G})-\underline\iota_{t}\big(e_{t}-(\psi_{t}\widetilde{d}_{t}+p_{t})/{\eta}\big)\nonumber\\
&\hspace{-0.2cm}+\overline\iota_{t}\big(e_{t}-\overline{E}+(b_{t}-\psi_{t}\widehat{d}_{t})\eta\big)\nonumber \\
&\hspace{-0.2cm} \dfrac{\partial L}{\partial g_{t}}=\dfrac{\partial \mathbb{E}\sum\nolimits_{t}{{G}}\left( {{g}_{t}}+\varphi_{t}\bm{{d}_{t}} \right)}{\partial {{g}_{t}}}-\lambda_{t}-\underline\nu_{t}+\overline\nu_{t}=0\label{lg2}\\
&\hspace{-0.2cm} \dfrac{\partial L}{\partial b_{t}}=-{{\theta }_{t}}{{\eta }}+\lambda_{t}-\underline{\alpha}_{t}+\overline{\alpha}_{t}+\overline\iota_{t}\eta=0\label{dc2}\\
& \hspace{-0.2cm}\dfrac{\partial L}{\partial p_{t}}=M+{{\theta }_{t}}/{\eta }-\lambda_{t}-\underline{\beta }_{t}+\overline{\beta }_{t}+\underline\iota_{t}/\eta=0\label{dd2}\\
&\hspace{-0.2cm} \dfrac{\partial L}{\partial e_{t}}=-{{\theta }_{t}}+{{\theta }_{t-1}}-\underline\iota_{t}+\overline\iota_{t}=0\label{dre2}\\
& \hspace{-0.2cm}\dfrac{\partial L}{\partial \varphi_{t}}=\dfrac{\partial \mathbb{E}\sum\nolimits_{t}{{G}}\left( {{g}_{t}}+\varphi_{t}\bm{{d}_{t}} \right)}{\partial \varphi_{t}}-\pi_{t}-\underline\nu_{t}\widehat{d}_{t}+\overline\nu_{t}\widetilde{d}_{t}=0\label{dg2}\\
& \hspace{-0.2cm}\dfrac{\partial L}{\partial \psi_{t}}=M\mu_{t}-\pi_{t}-\overline\alpha_{t}\widehat{d}_{t}+\overline\beta_{t}\widetilde{d}_{t}+\underline\iota_{t}\widetilde{d}_{t}/\eta-\overline\iota_{t}\widehat{d}_{t}\eta=0\label{de2}
\end{align}\label{KKT}\vspace{-0.5cm}
\end{subequations}

\subsection{Proof of Proposition 2} \label{appendix1}

\begin{proof} 
The conventional generator capacity is normally sufficient large, so that the constraints~\eqref{rgulbound} are not binding. Hence, we have $\underline\nu_{t}=\overline\nu_{t}=0$, then the KKT conditions~\eqref{lg2}-\eqref{dd2} can be simplified as:
\begin{subequations}
\begin{align}
& \partial L/\partial g_{t}=\partial \mathbb{E}\sum\nolimits_{t}{{G}}\left( {{g}_{t}}+\varphi_{t}\bm{{d}_{t}} \right)/\partial {{g}_{t}}-\lambda_{t}=0\label{dg}\\
& \partial L/\partial b_{t}=-{{\theta }_{t}}{{\eta }}+\lambda_{t}-\underline{\alpha}_{t}+\overline{\alpha}_{t}+\overline\iota_{t}\eta=0\label{dc}\\
& \partial L/\partial p_{t}=M+{{\theta }_{t}}/{\eta }-\lambda_{t}-\underline{\beta }_{t}+\overline{\beta }_{t}+\underline\iota_{t}/\eta=0\label{dd}
\end{align}\label{KKT-1}
\end{subequations}
To simplify the notation, we use the following substitution: $\bm{{D}_{t}}=D_{t}+\varphi_{t}\bm{{d}_{t}}$, $\partial \mathbb{E} {{G}}\left( {{g}_{t}}+\varphi_{t}\bm{{d}_{t}} \right)/{ \partial {{g}_{t}}}=H({{g}_{t}}+\varphi_{t}\bm{{d}_{t}})$, where $H$ is a monotomically increasing and linear function. Considering different binding conditions for storage power. We provide the following proof
by analyzing five cases:

\textbf{Case 1:} If constraint~\eqref{pcbound} is not binding, i.e., $b_{t}\neq0$ \& $b_{t}\neq \overline{P}$, then we have $p_{t}=0$, $\underline{\alpha}_{t}=\overline{\alpha}_{t}=\overline{\beta }_{t}=\overline\iota_{t}=0$. By combining~\eqref{dg}-\eqref{dd} and~\eqref{pb}, we can draw the following conclusion: 
\begin{subequations}\label{case1}
\begin{align}
& {{\theta }_{t}}=H({{g}_{t}}+\varphi_{t}\bm{{d}_{t}})/{\eta}=H\big(\bm{{D}_{t}}+(e_{t+1}-e_{t})/{\eta }\big)/{\eta}\\
& \dfrac{\partial{{\theta }_{t}}}{\partial{e_{t}}}=-\dfrac{1}{{{\eta }}^{2}}\dfrac{\partial H\big( \bm{{D}_{t}}+(e_{t+1}-e_{t})/{\eta }\big)}{\partial \big( \bm{{D}_{t}}+(e_{t+1}-e_{t})/{\eta }\big)}\leq 0 
\end{align}
\end{subequations}

\textbf{Case 2:} If constraint~\eqref{pdbound}~is not binding, i.e., $p_{t}\neq0$ \& $p_{t}\neq \overline{P}$, then we have $b_{t}=0$, $\overline{\alpha}_{t}=\underline{\beta }_{t}=\overline{\beta }_{t}=\underline\iota_{t}=0$. By combining~\eqref{dg}-\eqref{dd} and~\eqref{pb}, we can draw the following conclusion: 
\begin{subequations}
\begin{align}
& {{\theta }_{t}}+{{\eta }}M={{\eta }}H({{g}_{t}}+\varphi_{t}\bm{{d}_{t}})={{\eta }}H\big( {\bm{D}_{t}}+(e_{t+1}-e_{t}){\eta }\big)\\
& \dfrac{\partial{{\theta }_{t}}}{\partial{e_{t}}}=-{{\eta }}^{2}\dfrac{{\partial } H\big( {\bm{D}_{t}}+(e_{t+1}-e_{t}){\eta }\big)}{\partial \big( {\bm{D}_{t}}+(e_{t+1}-e_{t}){\eta }\big)}\leq 0 
\end{align}
\end{subequations}

\textbf{Case 3:} If $p_{t}=\overline{P}$, then we have $b_{t}=0$, $\overline{\alpha}_{t}=\underline{\beta }_{t}=\overline\iota_{t}=0$. And $\underline{\alpha}_{t},\overline{{\beta }_{t}}$ are positive and bounded. By combining~\eqref{dg}-\eqref{dc} and~\eqref{pb}, we can draw the upper bound of opportunity price:
\begin{subequations}
\begin{align}
&  \sup({{\theta }_{t}})=H\left( {{g}_{t}}+\varphi_{t}\bm{{d}_{t}} \right)/{ \eta}=H\big( {\bm{D}_{t}}+(e_{t+1}-e_{t}){\eta }\big)/{\eta}\\
& \dfrac{\partial{\sup({\theta }_{t}})}{\partial{e_{t}}}=-\dfrac{{\partial } H\big( {\bm{D}_{t}}+(e_{t+1}-e_{t}){\eta }\big)}{\partial \big( {\bm{D}_{t}}+(e_{t+1}-e_{t}){\eta }\big)}\leq 0 
\end{align}
\end{subequations}

\textbf{Case 4:} If $b_{t}=\overline{P}$, then we have $p_{t}=0$, $\underline{\alpha}_{t}=0$, $\overline{{\beta }_{t}}=\underline{\iota}_{t}=0$. And $\overline{\alpha}_{t},\underline{{\beta }_{t}}$ are positive and bounded. By combining~\eqref{dg}, \eqref{dd} and~\eqref{pb}, we can draw the lower bound of opportunity price:
\begin{subequations}
\begin{align}
& \hspace{-0.4cm} \inf({{\theta }_{t}}+{{\eta }}M)={\eta }H({{g}_{t}}+\varphi_{t}\bm{{d}_{t}})={{\eta }}H( {\bm{D}_{t}}+(e_{t+1}-e_{t})/{\eta })\\
& \hspace{-0.4cm} \dfrac{\partial{\inf({\theta }_{t}})}{\partial{e_{t}}}=-\dfrac{{\partial } H( {\bm{D}_{t}}+(e_{t+1}-e_{t})/{\eta })}{\partial ( {\bm{D}_{t}}+(e_{t+1}-e_{t})/{\eta })}\leq 0
\end{align}
\end{subequations}

\textbf{Case 5:} If $b_{t}=p_{t}=0$, then we have $\overline{\alpha}_{t}=0$, $\overline{{\beta }_{t}}=0$. And $\underline{\alpha}_{t},\underline{{\beta }_{t}}$ are positive and bounded. By combining~\eqref{dg}-\eqref{dd} and~\eqref{pb}, we can draw the following conclusion:
\begin{subequations}
\begin{align}
&\hspace{-0.3cm} \sup({{\theta }_{t}}-\overline{\iota}_{t})=H({{g}_{t}}+\varphi_{t}\bm{{d}_{t}})/{\eta}=H( {\bm{D}_{t}}+e_{t+1}-e_{t})/{\eta}\\
&\hspace{-0.3cm} \inf({{\theta }_{t}}+{{\eta }}M+\underline{\iota}_{t})={{\eta }}H({{g}_{t}}+\varphi_{t}\bm{{d}_{t}})={{\eta }}H( {\bm{D}_{t}}+e_{t+1}-e_{t})\\
&\hspace{-0.3cm} \partial{\inf({\theta }_{t}})/\partial{e_{t}}\leq0\text{, }\partial{\sup({\theta }_{t}})/\partial{e_{t}}\leq0 
\end{align}
\end{subequations}
Hence, we have finished the proof.
\end{proof}

\subsection{Proof of Theorem 1}\label{appendix2}

\begin{proof}
We proceed the following proof
by analyzing five cases with different binding conditions for constraints (1h-1i):

\textbf{Case 1:} If constraint~\eqref{pcbound} is not binding, i.e., $b_{t}\neq0$ \& $b_{t}\neq \overline{P}$, then we have $p_{t}=0$, $\underline{\alpha}_{t}=\overline{\alpha}_{t}=\overline{{\beta }_{t}}=0$. By combining~\eqref{dc2}\eqref{dre2}\eqref{de2}, we can draw the conclusion~\eqref{reserver1}.
\begin{equation}
\theta_{t-1}=\dfrac{\eta}{\widetilde{d}_{t}}\big(\theta_{t}\eta\widehat{d}_{t}+\lambda_{t}(\dfrac{\widetilde{d}_{t}}{\eta^2}-\widehat{d}_{t})+\pi_{t}-M\mu_{t} \big)\label{reserver1} 
\end{equation}

\textbf{Case 2:} If constraint~\eqref{pdbound}~is not binding, i.e., $p_{t}\neq0$ \& $p_{t}\neq \overline{P}$, then we have $b_{t}=0$, $\overline{\alpha}_{t}=\underline{{\beta }_{t}}=\overline{{\beta }_{t}}=0$. By combining~\eqref{dd2}\eqref{dre2}\eqref{de2}, we can draw the conclusion~\eqref{reserver2}. 
\begin{equation}
\theta_{t-1}=\dfrac{1}{\eta\widehat{d}_{t}}\big(\dfrac{\theta_{t}\widetilde{d}_{t}}{\eta}+\lambda_{t}(\eta^2\widehat{d}_{t}-\widetilde{d}_{t})+\pi_{t}+M(\widetilde{d}_{t}-\eta^2\widehat{d}_{t}-\mu_{t})  \big)\label{reserver2} 
\end{equation}

\textbf{Case 3:} If $p_{t}=\overline{P}$, then we have $b_{t}=0$, $\overline{\alpha}_{t}=\underline{{\beta }_{t}}=0$. And $\underline{\alpha}_{t},\overline{{\beta }_{t}}$ are positive and bounded. By combining~\eqref{dc2}\eqref{dre2}\eqref{de2}, we can draw the conclusion~\eqref{reserver3}.
\begin{equation}
\sup(\theta_{t-1})=\dfrac{\eta}{\widetilde{d}_{t}}\big(\theta_{t}\eta\widehat{d}_{t}+\lambda_{t}(\dfrac{\widetilde{d}_{t}}{\eta^2}-\widehat{d}_{t})+\pi_{t}-M\mu_{t}  \big)\label{reserver3} 
\end{equation}

\textbf{Case 4:} If $b_{t}=\overline{P}$, then we have $p_{t}=0$, $\underline{\alpha}_{t}=\overline{{\beta }_{t}}=0$. And $\overline{\alpha}_{t},\underline{{\beta }_{t}}$ are positive and bounded. By combining~\eqref{dd2}\eqref{dre2}\eqref{de2}, we can draw the conclusion~\eqref{reserver4}. 
\begin{equation}
\inf(\theta_{t-1})=\dfrac{1}{\eta\widehat{d}_{t}}\big(\dfrac{\theta_{t}\widetilde{d}_{t}}{\eta}+\lambda_{t}(\eta^2\widehat{d}_{t}-\widetilde{d}_{t})+\pi_{t}+M(\widetilde{d}_{t}-\eta^2\widehat{d}_{t}-\mu_{t}) \big)\label{reserver4} 
\end{equation}

\textbf{Case 5:} If $b_{t}=p_{t}=0$, then we have $\overline{\alpha}_{t}=\overline{{\beta }_{t}}=0$. And $\underline{\alpha}_{t},\underline{{\beta }_{t}}$ are positive and bounded. By combining~\eqref{dc2}\eqref{dd2}\eqref{dre2}\eqref{de2}, we can draw the conclusion~\eqref{reserver5}.  
\begin{subequations}
\begin{align}
    & \theta_{t-1}\leq\dfrac{\eta}{\widetilde{d}_{t}}\big(\theta_{t}\eta\widehat{d}_{t}+\lambda_{t}(\dfrac{\widetilde{d}_{t}}{\eta^2}-\widehat{d}_{t})+\pi_{t}-M\mu_{t} \big)\\
& \theta_{t-1}\geq\dfrac{1}{\eta\widehat{d}_{t}}\big(\dfrac{\theta_{t}\widetilde{d}_{t}}{\eta}+\lambda_{t}(\eta^2\widehat{d}_{t}-\widetilde{d}_{t})+\pi_{t}+M(\widetilde{d}_{t}-\eta^2\widehat{d}_{t}-\mu_{t})  \big)
\end{align}\label{reserver5}
\end{subequations}
Hence, we have finished the proof.
\end{proof}

\subsection{Formulation of Storage Profit Maximization} \label{appendix3}

\textbf{(1) Storage Profit Maximization.} Storage participants will design their bids using a profit-maximization model based on a set of energy price predictions $\lambda_{t}$. To handle the SoC dependencies in the storage
model, dynamic programming is adopted in~\cite{qin2024economic}, and the opportunity value function of storage can be calculated recursively as~\eqref{value}. Constraints~\eqref{se}-\eqref{eb} limit the power and SoC of storage. Constraint~\eqref{np} prevents energy storage from discharging at negative prices.
\begin{subequations}
\begin{align}
    & V_{t-1}({e}_{t-1})=\max _{p_t\text{, }b_t}\ \lambda_t\left(p_t-b_t\right)-M p_t+V_t({e}_{t})\label{value}\\
& {e}_{t+1}-{e}_{t}=-p_{t}/{\eta}+b_{t}{{\eta }}\label{se}\\
& 0\le b_{t}\le \overline{P}\text{, }0\le p_{t}\le \overline{P}\\
&0\le{e}_{t}\le\overline{E}\label{eb}\\
&  p_{t}=0\text{, }\lambda_t<0\label{np} 
\end{align}\label{profitmax}\vspace{-0.5cm}
\end{subequations}

\noindent where $V_{t}$ is the opportunity value of the
energy storage at time period $t$, hence the
value-to-go function in dynamic programming. And $V_t$ is dependent on SoC. 

\textbf{(2) Storage Bids Based on Opportunity Value Function.} Given the calculated storage opportunity value functions, the storage can generate charge and discharge bids based on the subderivatives of the physical cost and opportunity value functions as~\eqref{valuefunction}. 
\begin{subequations}
\begin{align}
 &O_t\left(p_t\right) =\dfrac{\partial(Mp_{t}+V_t(e_{t}))}{\partial p_{t}}=M+\frac{1}{\eta} v_t\left(e_{t-1}-p_t / \eta\right) \\ 
 &B_t\left(b_t\right) =\dfrac{\partial(V_t(e_{t}))}{\partial b_{t}}=\eta v_t\left(e_{t-1}+b_t \eta\right)
\end{align}\label{valuefunction}\vspace{-0.3cm}
\end{subequations}

\noindent where $O_t$ and $B_t$ are discharge and charge bids for energy storage. $v_t$ is the subderivative of $V_t$.

\textbf{(3) Market Clearing with Storage Bids.} System operator clear the market based on an economic dispatch problem as~\eqref{multiED-profmax}.
\begin{subequations}
\begin{align}
 &\min\ \sum\nolimits_{i} \sum\nolimits_{t}\ {{{G}}\left( {{g}_{i,t}}\right)}+O_t\left(p_t\right)-B_t\left(b_t\right)\\ 
 s.t.\ &\sum\nolimits_{i} \varphi_{t}=1\text{, }\varphi_{t}\leq 1\\
& 0\le b_{t}\le \overline{P}\text{, }0\le p_{t}\le \overline{P}\\
&0\le{e}_{t}\le\overline{E}\\
& \eqref{pb}-\eqref{SoC}\text{, }\eqref{rgulbound}-\eqref{st} 
\end{align}\label{multiED-profmax}
\end{subequations}


\end{document}